\documentclass[authoryear]{elsarticle}
\usepackage{amsmath,amssymb,graphicx,amsthm}
\usepackage{natbib}
\usepackage{color}
\usepackage{soul}

\newcommand{\D}{{\mathrm{d}}}
\newtheorem{theorem}{Theorem}

\newtheorem{proposition}{Proposition}
\newtheorem{definition}{Definition}
\newtheorem{corollary}{Corollary}
\newtheorem{remark}{Remark}

\journal{Chemical Engineering Science}

\begin{document}

\begin{frontmatter}

\title{Extended Detailed Balance for Systems \\  with Irreversible Reactions}
\author{A.~N.~Gorban}
 \ead{ag153@le.ac.uk}
\address{Department of Mathematics, University of Leicester, Leicester, LE1 7RH, UK}
\author{G.~S.~Yablonsky}
\address{Parks College, Department of Chemistry, Saint Louis
University, Saint Louis, MO 63103, USA}



\begin{abstract}
The principle of detailed balance states that in equilibrium each
elementary process is equilibrated by its reverse process. For many
real physico-chemical complex systems (e.g. homogeneous combustion,
heterogeneous catalytic oxidation, most enzyme reactions etc),
detailed mechanisms include both reversible and irreversible
reactions. In this case, the principle of detailed balance cannot be
applied directly. We represent irreversible reactions as limits of
reversible steps and obtain the principle of detailed balance for
complex mechanisms with some irreversible elementary processes. We
prove two consequences of the detailed balance for these mechanisms:
the structural condition and the algebraic condition that form
together the {\em extended form of detailed balance}. The {\em
algebraic condition} is the principle of detailed balance for the
reversible part. The {\em structural condition} is: the convex hull
of the stoichiometric vectors of the irreversible reactions has
empty intersection with the linear span of the stoichiometric
vectors of the reversible reactions. Physically, this means that the
irreversible reactions cannot be included in oriented cyclic
pathways.

The systems with the extended form of detailed balance are also
the limits of the reversible systems with detailed balance when
some of the equilibrium concentrations (or activities) tend to
zero. Surprisingly, the structure of the limit reaction
mechanism crucially depends on the relative speeds of this
tendency to zero.
\end{abstract}

\end{frontmatter}

\section{Introduction\label{sec1}}

\subsection{Detailed Balance for Systems
with Irreversible Reactions: the Grin of the Vanishing Cat}

The principle of detailed balance was explicitly introduced and
effectively used for collisions by \cite{Boltzmann}. In 1872,
he proved his $H$-theorem using this principle. In its general
form, this principle is formulated for kinetic systems which
are decomposed into elementary processes (collisions, or steps,
or elementary reactions). At equilibrium, each elementary
process should be equilibrated by its reverse process. The
arguments in favor of this property are founded upon
microscopic reversibility. The microscopic ``reversing of time"
turns at the kinetic level into the ``reversing of arrows": the
elementary processes transform into their reverse processes.
For example, the reaction $\sum_i \alpha_i A_i \to \sum_j
\beta_j B_j$ transforms into $\sum_j \beta_j B_j \to \sum_i
\alpha_i A_i$ and conversely. The equilibrium ensemble should
be invariant with respect to this transformation because of
microreversibility and the uniqueness of thermodynamic
equilibrium. This leads us immediately to the concept of
detailed balance: each process is equilibrated by its reverse
process.

For a given equilibrium, the principle of detailed balance
results in a system of {\em linear} conditions on kinetic
constants (or collision kernels). On the contrary, if we
postulate just the {\em existence} of an a priori unknown
equilibrium state with the detailed balance property then a
system of {\em nonlinear} conditions on kinetic constants
appear. These conditions were introduced in by
\cite{Wegscheider1901} and used later by \cite{Onsager1_2}.
They are known now as the {\em Wegscheider conditions}.

For linear kinetics, the Wegscheider conditions have a very
simple and transparent form: {\em for each oriented cycle of
elementary processes the product of kinetic constants is equal
to the product of kinetic constants of the reverse processes.}

However, many mechanisms of complex chemical and biochemical
reactions, in particular mechanisms of combustion and enzyme
reaction, include some irreversible (unidirectional) reactions.
In many cases, complex mechanisms consist of some reversible
and some irreversible reactions, equilibrium concentrations and
rates of reactions become zeroes, and the standard forms of the
detailed balance do not have a sense.

{In physical chemistry, the feasibility of a reaction depends on the
energies and entropies of system states, initial, final, and
transition ones. Nevertheless, some combinations of irreversible
reactions are impossible irrespective of the values of thermodynamic
functions. Since Wegscheider's time it is known that the cyclic
sequence of irreversible reactions (the completely irreversible
cycle) is impossible. It is forbidden by the principle of detailed
balance. In a similar way, the reaction mechanism} $A
\rightleftharpoons B$, $A \to C$, $C \to B$ is forbidden as well as
$A \rightleftharpoons B$, $A \rightleftharpoons  C$, $ C \to B$.

Two fundamental problems can be posed:

(1)   Which mechanisms with irreversible steps are allowed, and
which such mechanisms are forbidden by the principle of
detailed balance?

In accordance with our knowledge, this question was not
answered rigorously and the general problem was not solved.
Beside that, the procedure of determining the forbidden
mechanisms was not described.

(2) Let a mechanism with some irreversible steps be not forbidden.
Do we still have some relationships between kinetic constants of
this mechanism?

In our paper, both problems are analyzed based on the same
procedure. Substituting the zero kinetic constants by small, however
not zero values we return to the fully 'reversible case', in which
all steps of the reaction mechanism are reversible. Then, we analyze
a limit case, in which small kinetic parameters tend to reach 0.

Such an idea was applied previously to several examples. In
particular, \cite{Chu1971} used this idea for a three-step
mechanism, demonstrating that the mechanism $A \rightleftharpoons
B$, $A \to C$, $B \to C$ can appear as a limit of reversible
mechanisms which obey the principle of detailed balance, whereas the
system $A \rightleftharpoons B$, $A \to C$, $C \to B$ cannot appear
in such a limit. However, this approach was not applied to the
general analysis of multi-step mechanisms, only to a few systems of
low dimensions.

Since Lewis Carroll's ``Alice's Adventures in Wonderland", the
Cheshire Cat is well known, in particular its inscrutable grin.
Finally this cat disappears gradually until nothing is left but
its grin. Alice makes a remark she has often seen a cat without
a grin but never a grin without a cat.

The detailed balance for systems with irreversible reactions
can be compared with this grin of the Cheshire cat: the whole
cat (the reversible system with detailed balance) vanishes but
the grin persists.

\subsection{Detailed Balance: the Classical Relations
\label{sec:Detbal}}

First, let us consider linear systems and write the general
first order kinetic equations:
\begin{equation}\label{Master1}
\dot{p}_i=\sum_j (k_{ij}p_j-k_{ji}p_i) \, .
\end{equation}
Here, $p_i$ is the probability of a state $A_i$
($i=1,\ldots,n$) (or, for monomolecular reactions, the
concentration of a reagent $A_i$). The kinetic constant $k_{ij}
\geq 0$ ($i \neq j$) is the intensity of the transitions
$A_j\to A_i$ (i.e., $k_{ij}$ is $k_{i \leftarrow j}$). The rate
of the elementary process $A_i \to A_j$ is $k_{ji} p_i$. The
class of equations (\ref{Master1}) includes the Kolmogorov
equation for finite Markov chains, the Master equation in
physical kinetics and the chemical kinetics equations for
monomolecular reactions.

Let $p_i^{\rm eq}>0$ be a positive equilibrium distribution.
According to  the principle of detailed balance, the rate of
the elementary process $A_i \to A_j$ at equilibrium coincides
with the rate of the reverse process $A_i \leftarrow A_j$:
\begin{equation}\label{DetBalMastLin}
k_{ij}p_j^{\rm eq}=k_{ji}p_i^{\rm eq} \, .
\end{equation}
For a given equilibrium, $p_i^{\rm eq}$, the principle of
detailed balance is equivalent to this system
(\ref{DetBalMastLin}) of linear equalities. To find the
conditions of the {\em existence} of such a positive
equilibrium that (\ref{DetBalMastLin}) holds, it is sufficient
to write equations (\ref{DetBalMastLin}) in the logarithmic
form, $\ln p_i^{\rm eq} - \ln p_j^{\rm eq} = \ln k_{ij}-\ln
k_{ji}$, to consider this system as a system of linear
equations with respect to the unknown $\ln p_i^{\rm eq}$, and
to formulate the standard solvability condition.

After some elementary transformation this condition gives: a
positive equilibrium with detailed balance (\ref{DetBalMastLin})
exists if and only if
\begin{enumerate}
\item{If $k_{ij}>0$ then $k_{ji}>0$ (reversibility);}
\item{For each oriented cycle of elementary processes,
    $A_{i_1}\to A_{i_2} \to \ldots A_{i_q}\to A_{i_1}$, the
    product of the kinetic constants is equal to the
    product of the kinetic constants of the reverse processes:
\begin{equation}\label{WegscheiderCondition}
\prod_{j=1}^q k_{i_{j+1}i_{j}}= \prod_{j=1}^q k_{i_{j}i_{j+1}}
\end{equation}
where the cyclic numeration is used, $i_{q+1}=i_1$.}
\end{enumerate}
Of course, it is sufficient to use in
(\ref{WegscheiderCondition}) a basis of independent cycles
(see, for example the review of \cite{Schnakenberg1976}).

Let us introduce the more general Wegscheider conditions for
nonlinear kinetics and the generalized mass action law. (For a
more detailed exposition we refer to the textbook of
\cite{Yablonskiiatal1991}.) The elementary reactions are given
by the stoichiometric equations
\begin{equation}\label{ReactMech}
\sum_i \alpha_{ri} A_i \to \sum_j \beta_{rj} A_j \;\; (r=1, \ldots, m) \, ,
\end{equation}
where $A_i$ are the components and $\alpha_{ri}\geq 0$,
$\beta_{rj}\geq 0$ are the stoichiometric coefficients. The reverse
reactions with positive constants are included in the list
(\ref{ReactMech}) separately. We need this separation of direct and
reverse reactions to apply later the general formalism to the
systems with some irreversible reactions.

The {\em stoichiometric matrix} is
$\boldsymbol{\Gamma}=(\gamma_{ri})$,
$\gamma_{ri}=\beta_{ri}-\alpha_{ri}$ (gain minus loss). The
{\em stoichiometric vector} $\gamma_r$ is the $r$th row of
$\boldsymbol{\Gamma}$ with coordinates
$\gamma_{ri}=\beta_{ri}-\alpha_{ri}$.

According to the {\em generalized mass action law}, the
reaction rate for an elementary reaction (\ref{ReactMech}) is
\begin{equation}\label{GenMAL}
w_r=k_r \prod_{i=1}^n a_i^{\alpha_{ri}} \, ,
\end{equation}
where $a_i\geq 0$ is the {\em activity} of $A_i$.

The list (\ref{ReactMech}) includes reactions with the reaction
rate constants $k_r>0$. For each $r$ we define $k_r^+=k_r$,
$w_r^+=w_r$, $k_r^-$ is the reaction rate constant for the
reverse reaction if it is on the list (\ref{ReactMech}) and 0
if it is not, $w_r^-$ is the reaction rate for the reverse
reaction if it is on the list (\ref{ReactMech}) and 0 if it is
not. For a reversible reaction, $K_r=k_r^+/k_r^-$

The principle of detailed balance for the generalized mass
action law is: For given values $k_r$ there exists a positive
equilibrium $a_i^{\rm eq}>0$ with detailed balance,
$w_r^+=w_r^-$. This means that the system of linear equations
\begin{equation}\label{DetBalGen}
\sum_i \gamma_{ri} x_i = \ln k_r^+-\ln k_r^-=\ln K_r
\end{equation}
is solvable ($x_i=\ln a_i^{\rm eq}$). The following classical
result gives the necessary and sufficient conditions for the
existence of the positive equilibrium $a_i^{\rm eq}>0$ with
detailed balance (see, for example, the textbook of
\cite{Yablonskiiatal1991}).

\begin{proposition}
Two conditions are sufficient and necessary for solvability of
(\ref{DetBalGen}):
\begin{enumerate}
\item{If $k_r^+>0$ then $k_r^->0$ (reversibility);}
\item {For any solution $ \boldsymbol{\lambda}=
    (\lambda_r)$ of the system
\begin{equation}\label{lambdaGamma}
\boldsymbol{\lambda \Gamma} =0  \;\; \left(\mbox{i.e.}\;\; \sum_r \lambda_r \gamma_{ri}=0\;\; \mbox{for all} \;\; i\right)
\end{equation}
the Wegscheider identity holds:
\begin{equation}\label{WegscheiderLambda}
\prod_{r=1}^m     (k_r^+)^{\lambda_r}=\prod_{r=1}^m     (k_r^-)^{\lambda_r} \, .
\end{equation}}
\end{enumerate}
\end{proposition}

\begin{remark}It is sufficient to use in (\ref{WegscheiderLambda}) a
basis of solutions of the system (\ref{lambdaGamma}):
$\boldsymbol{\lambda} \in \{\boldsymbol{\lambda}^1, \cdots ,
\boldsymbol{\lambda}^g\}$.
\end{remark}

\begin{remark}
The Wegscheider condition for the linear systems
(\ref{WegscheiderCondition}) is a particular case of the
general Wegscheider identity (\ref{WegscheiderLambda}).
Therefore, the solutions $ \boldsymbol{\lambda}$ of equation
(\ref{lambdaGamma}) are generalizations of the (non-oriented)
cycles in the reaction networks. The basis of solutions
corresponds to the basic cycles. This basis is, obviously, not
unique.
\end{remark}

\begin{remark}
In equation (\ref{DetBalGen}) unknown $x_i=\ln a_i$ are independent
variables and vector $\boldsymbol{x}$ can take any value in $R^n$.
In practice, this is not always true. For example, for heterogeneous
systems with solid components some activities may vary in a narrow
interval or may be even constant (see the more detailed discussion
below in Section \ref{Sec:Methane}). We do not study multiphase
equilibiria in our paper.
\end{remark}

\begin{remark}
All the closed chemical systems have linear conservation laws:
conservation of mass, various sorts of atoms, electric charge and
other conserved quantities. They are linear functions of the amounts
$N_i$ of chemical components $A_i$.  There is a problem of
uniqueness and existence of a positive equilibrium with detailed
balance or without it for every set of values of the independent
conservation laws. To solve this problem we need some properties of
the connection between activities and concentrations,
$(a_i)\leftrightarrow(c_i)$. We do not assume any hypothesis about
this connection and study just existence of a positive equilibrium
with detailed balance in the space of activities. The Wegscheider
identity (\ref{WegscheiderLambda}) gives a necessary and sufficient
condition for this existence.
\end{remark}

In practice, very often $k_r^-=0$ for some $r$, whereas
$k_r^+>0$. In these cases, the standard forms of the detailed
balance have no sense. Indeed, let us consider a linear
reversible cycle with an irreversible buffer:
$$A_1 \rightleftharpoons A_2 \rightleftharpoons \ldots A_n
\rightleftharpoons A_1 \to A_0 \, .$$ This system converges to
the state where only $p_0>$ and $p_i=0$ for $i>0$. In this
state, trivially, $w_r^+=w_r^-=0$ and it seems that the
standard principle of detailed balance does not imply any
restriction on the kinetic constants. Of course, this
impression is wrong.

Let us consider this system as a limit of the system with a
reversible buffer, $A_1 \rightleftharpoons A_0 $ (both reaction
rate constants are positive), when the constant of the reverse
reaction is positive but tends to zero: $k_{1 \leftarrow 0} \to
0$, $k_{1 \leftarrow 0}> 0$. For each positive value $k_{1
\leftarrow 0}> 0$ the  condition of detailed balance
$w_r^+=w_r^-$ gives the Wegscheider identity
(\ref{WegscheiderCondition}) for the cycle $A_1
\rightleftharpoons A_2 \rightleftharpoons \ldots A_n
\rightleftharpoons A_1$: The product of direct reaction rate
constants is equal to the product of the reverse reaction rate
constants. This condition holds also in the limit $k_{1
\leftarrow 0} \to 0$. So, any practically negligible but
positive value of the reverse kinetic constant implies the
nontrivial Wegscheider condition for the other constants.

If we assume that the negligible values of the constants should
not affect the kinetic systems then this Wegscheider condition
should hold for the system with fully irreversible steps as
well. Therefore, the following way for the formalization of the
principle of detailed balance for irreversible reactions is
proposed. We return to reversible reactions in which the
principle of detailed balance is assumed by the introduction of
small $k_r^->0$. Then we go to the limit $k_r^-\to 0$
($k_r^->0$) for these reactions.

{It is worth mentioning that the free energy has no limit when
some of the reaction equilibrium constants tend to zero. For
example, for the ideal gases $F=\sum_i N_i (RT\ln c_i +
\mu_i^0-RT)$, where $c_i$ is the concentration, $N_i$ is the
amount and $\mu_i^0$ is the standard chemical potential of the
component $A_i$. In the irreversible limit some $\mu_i^0 \to
\infty$. On the contrary, the activities}
\begin{equation}\label{StandardActivity}
a_i=\exp\left(\frac{\mu_i-\mu_i^0}{RT}\right)\,
\end{equation}
{remain finite (for the ideal gases, for example, $a_i=c_i$) and the
approach based on the generalized mass action law and the equations
$w_r^+=w_r^-$ can be applied in the irreversible limit.}

Below, we study systems with  irreversible reactions as the
limits of the systems with reversible reactions and detailed
balance, when some reaction rate constants go to zero. We
formulate the restrictions on the constants in this limit and
find the finite number of conditions that is necessary and
sufficient to check. First of all we have to discuss the
necessary notion of cycles for general reaction networks.

\subsection{Main Results}

We develop three approaches to the detailed balance conditions for
the systems with some irreversible reactions. The first and the most
physical idea is to consider an irreversible reaction as a limit of
a reversible reaction when the reaction rate constant for a reverse
reaction tends to zero. The limits of systems of reversible
reactions with detailed balance conditions cannot be arbitrary
systems with some irreversible reactions and we study the structural
and algebraic restrictions for these systems.

The second approach is based on the technical idea to study the
limits of the Wegscheider identities (\ref{WegscheiderLambda}).
Here, it is very useful to apply the concept of the general
(nonlinear) irreversible cycles or pathways developed recently far
enough for our purposes by
\cite{SchusterNature2000,ElemNodesBMS2004} and other. Let us write
all reactions separately (including direct and reverse reactions)
(\ref{ReactMech}). The general oriented cycle or pathway is a
relation between vectors $\gamma_r$ with non-negative coefficients :
$\sum_r \lambda_r \gamma_r=0$, $\lambda_r \geq 0$ and $\sum_r
\lambda_r >0$. For each system with all reversible reactions and
detailed balance the Wegscheider identity (\ref{WegscheiderLambda})
holds for any oriented cycle. Therefore, if an oriented cycle
persists in the limit with some irreversible reactions, then, for
$\lambda_r>0$, the $r$th reaction should remain reversible and for
this cycle the Wegscheider condition persists.

This property motivates the definition of the extended (or
weakened, \cite{YabGor2010}) form of detailed balance in
Section~\ref{Sec:DefWDB} through the general oriented cycles
and the Wegscheider conditions. Theorem~\ref{Theo:Equival}
states that a system satisfies the extended form of detailed
balance if and only if it is a limit of systems with all
reversible reactions and detailed balance. One part of this
theorem (``if") is proved immediately in
Section~\ref{Sec:DefWDB}, the proof of the second part (``only
if") exploits the third approach and is postponed till
Section~\ref{Sec:Shifted}.

The third idea is to study the limits when some equilibrium
concentrations (or, more general, activities) tend to zero. For
systems with all reversible reactions, we can explicitly express the
constants of the reverse reactions through the constants of the
direct reactions and the equilibrium activities: just use the
detailed balance conditions, $w^+(a^{\rm eq})=w^-(a^{\rm eq})$.
Here, instead of 2$m$ parameters, $k^{\pm}_r$ ($m$ is the number of
reactions) we use $m+n$ parameters: $m$ reaction rate constants
$k^+_r$ and $n$ equilibrium activities $a_i^{\rm eq}$. In this
description of the reversible reactions, the principle of detailed
balance is trivially satisfied. Some reactions become irreversible
in the limits when some of the equilibrium activities tends to zero.
Surprisingly, the structure of the limit reaction mechanism
crucially depends on the relative speeds of this tendency to zero.

In Section~\ref{Sec:Shifted},  we assume that $a_i^{\rm
eq}={\rm const_i} \times \varepsilon^{\delta_i}$ and study the
limits $\varepsilon \to 0$. The $n$-dimensional space of
exponents $\delta=(\delta_i)$ is split by $m$ hyperplanes
$(\gamma_r, \delta)=0$ on convex cones. Each of these cones is
given by a set of inequalities $(\gamma_r, \delta)\lesseqqgtr
0$ ($r=1,\ldots ,m$). In every such a cone, the limit reaction
mechanism for $\varepsilon \to 0$ is constant.

Using this approach, we prove the second part of
Theorem~\ref{Theo:Equival} and  even more: if a system satisfies the
extended form of detailed balance then it may be obtained in the
limit $\varepsilon \to 0$ from a system with all reversible
reactions with given $k_r^+$ and $a_i^{\rm eq}={\rm const_i} \times
\varepsilon^{\delta_i}$  for some exponents ${\delta_i}$
(Theorem~\ref{Theo:Exponents+Weak}). So, all the three approaches to
the consequences of the principle of detailed balance for the
systems with some irreversible reactions are equivalent.

The computational problem associated with the extended form of
detailed balance is nontrivial. For example, the oriented
cycles (pathways) form a convex polyhedral cone and we have to
formulate the structural condition of the extended form of
detailed balance for all extreme rays (extreme pathways) of
this cone (Theorem~\ref{Theo:2}): if  $\lambda_r>0$ for a
vector $ \boldsymbol{\lambda}$ from an extreme ray then the
$r$th reaction should remain reversible. Calculation of all
these extreme rays is a well known and computational expensive
problem \citep{FukudaProdon1996,Papin2003,ElemNodesBMS2004}. In
Theorem~\ref{Theo:ConstrCrit}, we significantly reduce the
dimension of the problem.

Instead of the set of {\em all} stoichiometric vectors $\gamma_r$
($r=1,\ldots, m$) in the whole space of composition $\mathbb{R}^n$
($n$ is the number of components, $m$ is the number of reactions) it
is sufficient to consider the set of the stoichiometric vectors of
the {\em irreversible} reactions in the quotient space
$\mathbb{R}^n/S$, where $S$ is spanned by the stoichiometric vectors
of all {\em reversible} reactions. The simple exclusion of the
linear conservation laws provides additional dimensionality
reduction.  The application of reduction methods is demonstrated in
the case study in Section~\ref{Sec:Methane}.

In Section~\ref{Sec:linear}, we formulate the main results for
the simple case of linear (monomolecular) systems.
Sections~\ref{Sec:NonlinExamples} and~\ref{Sec:Methane} are
devoted to examples of nonlinear systems. In
Section~\ref{Sec:NonlinExamples}, the simple examples with
obvious lists of the extreme pathways are collected. In
Section~\ref{Sec:Methane}, we analyze the possible irreversible
limits for a complex reaction of methane reforming with CO$_2$.

\section{Cycles and Pathways in General Reaction Networks}

There exist several graph representations of general reaction
networks \citep{Yablonskiiatal1991,Temkin1996} and each of them
implies the correspondent notion of a cycle. For example, each input
and output formal sum in the reaction mechanism (\ref{ReactMech})
can be considered as a vertex (a complex) and then a reaction with
the positive rate constant is an oriented edge. This graph of the
transformation of complexes is convenient for the analysis of the
complex balance condition \citep{Feinberg1972}.

The bipartite graphs of reactions \citep{VolpertKhudyaev1985}
gives us another example: it includes two types of vertices:
components (correspond to $A_i$) and reactions (correspond to
elementary reactions from (\ref{ReactMech})). There is an edge
from the $i$th component to the $s$th reaction if
$\alpha_{ri}>0$ and from the $s$th reaction to the $i$th
component if $\beta_{ri}>0$. The correspondent stoichiometric
coefficients are the weights of the edges. This graph is
convenient for the analysis of the system stability, for
calculation of Jacobians for the right hand sides of the
kinetic equations and for analysis of their signs
\citep{Ivanova1979,MinchevaRoussel2007}. For nonlinear systems,
these graphs do not give a simple representation of the
detailed balance conditions.

We need a special notion of a cycle which corresponds to the
algebraic relations between reactions. Let us recall that we
include direct and inverse reactions in the reaction mechanism
(\ref{ReactMech}) separately. Each solution of
(\ref{lambdaGamma}) may be represented as follows:
\begin{equation}\label{LambdaMech}
\begin{split}
+&\left|
\underline{\begin{array}{l}
\lambda_1 \times\left(\sum_i \alpha_{1i} A_i \to \sum_j \beta_{1j} A_j \right) \\
\lambda_2 \times\left(\sum_i \alpha_{2i} A_i \to \sum_j \beta_{2j} A_j \right) \\
\cdots            \\
\lambda_m \times\left(\sum_i \alpha_{mi} A_i \to \sum_j \beta_{mj} A_j \right)
\end{array}} \right. \\
& \;\;\;\;\;\;\; \;\;\;\;= \; \sum_i a_{i} A_i \to \sum_j a_{j} A_j \; .
\end{split}
\end{equation}
Here, $a_i=\sum_s \lambda_s \alpha_{si} \equiv \sum_s \lambda_s
\beta_{si}$. Therefore, we need the following definition of a
cycle.

\begin{definition}\label{Def:OrientedCycles}An oriented cycle is a vector of coefficients $\boldsymbol{\lambda}
\neq 0$ with all $\lambda_i \geq 0$ that satisfies
(\ref{LambdaMech}).
\end{definition}

\begin{remark}
Cycles in catalysis and, especially, in biochemistry  are
called pathways \citep{SchusterNature2000,Papin2003}. An
oriented pathway is an oriented cycle from
Definition~\ref{Def:OrientedCycles}. An extreme (oriented)
pathway is a direction vector of an extreme ray of the cone
$\Lambda_+$. A solution of equation (\ref{lambdaGamma}) (a
non-oriented cycle) is a non-oriented pathway.
\end{remark}

Qualitatively these concepts have been introduced in the early 1940s
by Horiuti who applied them to heterogeneous catalytic reactions
\citep{Horiuti1973}. Horiuti used them to eliminate intermediates of
the complex catalytic reaction by adding the steps of the detailed
mechanism first multiplied by special coefficients. As result of
such procedure, the chemical equation with no intermediates is
obtained.

All oriented cycles form the cone $\Lambda_+$ (without the
origin). {\em Extreme ray} of a convex cone is a face that is,
at the same time, a ray.  Each ray may be defined by a
directional vector $\boldsymbol{\lambda}$ that is an arbitrary
nonzero vector from this ray. The cone $\Lambda_+$ is defined
by a finite system of linear equations (\ref{lambdaGamma}) and
inequalities $\lambda_r \geq 0$. Therefore, it has a finite set
of extreme rays.

For integer stoichiometric coefficients, $\alpha_{si}$,
$\beta_{si}$, any extreme ray is defined by an uniform linear
systems of equations with integer coefficients supplemented by the
conditions $\lambda_i \geq 0$ and $\boldsymbol{\lambda} \neq 0$.
Therefore, we can always select a direction vector with the integer
coefficients. For each extreme ray, there exists a unique direction
vector with minimal integer coefficients.

For monomolecular reaction networks, these cycles coincide with
the oriented cycles in the graph of reactions (where vertices
are reagents and edges are reactions).

There exists an oriented cycle of the length two for each pair of
mutually reverse reactions. For these cycles the Wegscheider
identities (\ref{WegscheiderLambda}) hold trivially, for any
positive values of $k^{\pm}$.

\begin{remark}\label{Rem:Acycl}
The systems without oriented cycles ($\Lambda_+=\{0\}$) have a
simple dynamic behavior. First of all, for such a system the
convex hull of the stoichiometric vectors does not include
zero: $0 \notin {\rm conv}\{\gamma_1,\ldots,\gamma_m\}$.
Therefore, there exists a linear functional $l$ that separates
0 from $\{\gamma_1,\ldots,\gamma_m\}$: $l(\gamma_s) >0$ for all
$s=1, \ldots, m$. This linear function $l(c)$ increases
monotonically due to any kinetic equation $$\frac{\D c}{\D
t}=\sum_s w_s \gamma_s$$ with reaction rates $w_s \geq 0$: $\D
l(c)/\D t >0$ if at least one reaction rate $w_s>0$.
\end{remark}

\section{Extended Form of Detailed Balance}

\subsection{Definition \label{Sec:DefWDB}}

A practically important reaction  mechanism may include
reversible and irreversible steps. However, some mechanisms
with irreversible steps may be wrong because they cannot appear
as the limits of reversible mechanisms with detailed balance.
Therefore, the first question is about the mechanism structure:
what is allowed?

The second question is about the constants: let the mechanism
not be forbidden. If it is the limit of a system with detailed
balance then the reaction rate constants may be connected by
additional algebraic conditions like the Wegscheider conditions
(\ref{WegscheiderCondition}). We should describe all the
necessary conditions. In this Section we answer both questions
and formulate both conditions, structural and algebraic.

We have to study study the identities (\ref{WegscheiderLambda})
in the limit when some $k_r^-\to 0$. First of all, let us
consider reversible reactions: if $k_r^+ > 0$ then $k_r^- > 0$.
It is sufficient to use in (\ref{WegscheiderLambda}) only
$\boldsymbol{\lambda}$ with nonnegative coordinates, $\lambda_r
\geq 0$. Indeed, the direct and reverse reactions are included
in the list (\ref{ReactMech}) under different numbers. Assume
that $\lambda_r <0$ in an identity (\ref{WegscheiderLambda})
for some $r$. Let the reverse reaction for this $r$ have number
$r'$. Let us substitute $(k^+_r)^{\lambda_r}$ in  the left hand
side of (\ref{WegscheiderLambda}) by $(k^+_{r'})^{-\lambda_r}$
and $(k^-_r)^{\lambda_r}$ in  the right hand side by
$(k^-_{r'})^{-\lambda_r}$. The new condition is equivalent to
the previous one. Let us iterate this operation for various
$r$. In the finite number of steps all the powers $\lambda_r
\geq 0$.

Let us use notation $\Lambda$  for the linear space of
solutions of (\ref{lambdaGamma}) and $\Lambda_+$ for the cone
of positive solutions $\boldsymbol{\lambda}$ ($\lambda_r \geq 0
$) of (\ref{lambdaGamma}).

For reversible reactions, we proved the following proposition.
Let the reactions are reversible and the direct and reverse
reactions are included in the list (\ref{ReactMech})
separately.

\begin{proposition}\label{Prop:1}The Wegscheider identity  (\ref{WegscheiderLambda})
holds for all $\boldsymbol{\lambda}\in \Lambda$ if and only if
it holds for all positive $\boldsymbol{\lambda}\in \Lambda_+$.
\end{proposition}

Elementary linear algebra gives the following corollary for
reversible reactions.

\begin{corollary}\label{Corr:1}The solution of the system of linear equations for
logarithms of equilibrium activities (\ref{DetBalGen}) exists
if and only if for any {\em positive} solution
$\boldsymbol{\lambda}$ ($\lambda_r \geq 0 $) of the system
$\boldsymbol{\lambda} \boldsymbol{\Gamma}=0$
(\ref{lambdaGamma}) the condition (\ref{WegscheiderLambda})
holds.
\end{corollary}

Let us study  identity (\ref{WegscheiderLambda}) for a positive
$\boldsymbol{\lambda}$ when some of $k_r \to 0$. In this limit,
for every $\boldsymbol{\lambda} \in \Lambda_+$
Corollary~\ref{Corr:1} gives two conditions:

\begin{corollary}
Let $k_s>0$, $k_s \to k_s^{\rm lim}$ and the Wegscheider
identity (\ref{WegscheiderLambda}) holds for $k_s$. Then
\begin{enumerate}
\item{If  $\lambda_s >0$ and $k_s^+\to 0$ for some $s$ then for some $q$
    $\lambda_q  >0$ and $k_q^- \to 0$;}
\item{If for all positive components $\lambda_s >0$ the limit
    constants are positive, $k_s^{\rm lim\, \pm}>0$, then
    the condition (\ref{WegscheiderLambda}) holds for
    $k_s^{\rm lim\, \pm}$.}
\end{enumerate}
\end{corollary}

We can interpret the positive solutions of (\ref{lambdaGamma})
as oriented cycles (linear or nonlinear). The first condition
means that if a cycle is cut by the limit $k_s^+\to 0$ in one
direction then it should be also cut by a limit $k_q^-\to 0$ in
the opposite direction: the irreversible cycle is forbidden.
This remark leads to  the definition of the {\em structural
condition} of the extended form of detailed balance.

\begin{definition}
A system of reactions (\ref{ReactMech}) satisfies the
structural condition of the extended form of detailed balance
if for every $\boldsymbol{\lambda} \in \Lambda_+$ the reaction
which satisfy $\lambda_s>0$ are reversible: if $\lambda_s>0$
then $k_s^{\pm}>0$.
\end{definition}

This condition means that all cycles should be reversible. The
second condition means that for all cycles
$\boldsymbol{\lambda} \in \Lambda_+$ which persist in the
system with irreversible reactions the correspondent
Wegscheider condition (\ref{WegscheiderLambda}) holds. This is
the {\em algebraic condition} of the extended form of detailed
balance. Now, we are ready to formulate the definition of the
extended form of detailed balance.

\begin{definition}The subsystem satisfies the extended form of detailed balance if both the
structural and the algebraic condition hold for all
$\boldsymbol{\lambda} \in \Lambda_+$.
\end{definition}

The following theorem gives the motivation to this definition.

\begin{theorem}\label{Theo:Equival}A system with irreversible reactions is a limit of systems
with reversible reactions and detailed balance if and only if
it satisfies the extended form of detailed balance.
\end{theorem}

\begin{proof}
Let us prove the direct statement: if a system is a limit of systems
with reversible reactions and detailed balance then it satisfies the
extended form of detailed balance. Indeed, let a system of reactions
be a limit of systems with reversible reactions and detailed
balance. This means that for each $j=1,2,\ldots$ a set of reaction
rate constants $k_{s,j}^{\pm}>0$ is given, $k_{s,j}^{\pm}>0$ satisfy
the principle of detailed balance for all $j$ and
$$k_s^{\pm}=\lim_{j\to \infty} k_{s,j}^{\pm}\, .$$
Assume that  the structural condition is violated: there exists such
a $\boldsymbol{\lambda} \in \Lambda_+$ that $\lambda_s
>0$ for an irreversible reaction ($k_s^+>0$, $k_s^-=0$). For
all $j=1,2,\ldots$ the principle of detailed balance gives:
\begin{equation}\label{DetBalLim}
\prod_{r,\, \lambda_r>0} (k_{r,j}^+)^{\lambda_r} =\prod_{r,\,
\lambda_r>0} (k_{r,j}^-)^{\lambda_r}\, .
\end{equation}
If $\lambda_r>0$ then $k_r^+>0$. Therefore, for these $r$,
sufficiently large $j$ and some $\varepsilon, \delta>0$
$\delta>k_{r,j}^{\pm}>\varepsilon> 0$. The left hand side of
(\ref{DetBalLim}) is separated from zero. The right hand side
of (\ref{DetBalLim}) tends to zero  because all factors are
bounded and at least one of them tends to zero, $k_{r,j}^-\to
0$. This contradiction proves the structural condition. To
prove the algebraic condition, it is sufficient to notice that
the Wegscheider identity for $k_{s,j}^{\pm}>0$ holds for all
$j$, hence, it holds in the limit $j\to\infty$.

We will prove the converse statement (if a system  satisfies
the extended form of detailed balance then it is a limit of
systems with reversible reactions and detailed balance) in
Section~\ref{Sec:Shifted}, in the proof of
Theorem~\ref{Theo:Exponents+Weak}.
\end{proof}

\subsection{Criteria \label{Sec:Crit}}

All $\boldsymbol{\lambda} \in \Lambda_+$ participate in the
definition of the extended form of detailed balance.
Nevertheless, it is sufficient to use a finite subset of this
cone.

We can check directly that if for a set
$\{\boldsymbol{\lambda}^s\}$ the structural and the algebraic
conditions of the extended form of detailed balance hold then
they hold for any conic combination of
$\{\boldsymbol{\lambda}^s\}$, $\boldsymbol{\lambda}=\sum_s a_s
\boldsymbol{\lambda}^s$, $a_s \geq 0$. Therefore, it is
sufficient to check the conditions for the directional vectors
of the extreme rays of $\Lambda_+$.

Let a reaction mechanism satisfy the extended principle of
detailed balance. If we delete from this mechanisms any
irreversible elementary reaction or any couple of mutually
reverse elementary reactions, the resulting mechanism satisfies
the extended principle of detailed balance as well.

A  cone is pointed if the origin is its extreme point or, which
is the same, this cone does not include a whole straight line.
The cone $\Lambda_+$ is pointed because it belongs to the
positive orthant $\{\boldsymbol{\lambda} \ | \
\boldsymbol{\lambda}\geq 0\}$.

It is a standard task of linear programming and computational
convex geometry to find all the extreme rays  of the polyhedral
pointed cone $\Lambda_+$
\citep{Bertsimas1997,MotzkinAtAl1953,FukudaProdon1996}. Let the
directional vectors of these extreme rays be
$\{\boldsymbol{\lambda}^s \ | \ s=1, \cdots , q \}$.

\begin{theorem}\label{Theo:2}The system satisfies the extended form of detailed balance if
and only if the structural and algebraic conditions hold for
the directional vectors $\{\boldsymbol{\lambda}^s \ | \ s=1,
\cdots , q \}$ of the extreme rays of the cone $\Lambda_+$.
\end{theorem}

Theorem~\ref{Theo:2} follows just from the definition of
extreme rays and the Minkowski theorem which states that every
pointed cone given by linear inequalities admits a unique
representation as a conic hull of a finite set of extreme rays.

This criterion can be simplified as well: it is necessary and
sufficient to check the structural conditions for the extreme
rays of $\Lambda_+$ and then the algebraic condition for a
maximal linear independent subset of $\{\boldsymbol{\lambda}^s
\ | \ s=1, \cdots , q \}$.

\begin{corollary}The system satisfies the extended form of detailed balance if
and only if the structural conditions hold for all directional
vectors $\{\boldsymbol{\lambda}^s \ | \ s=1, \cdots , q \}$ of
the extreme rays of the cone $\Lambda_+$ and the algebraic
conditions hold for a maximal linear independent subset of
$\{\boldsymbol{\lambda}^s \ | \ s=1, \cdots , q \}$.
\end{corollary}

If, for a given reaction mechanism, the set
$\{\boldsymbol{\lambda}^s \ | \ s=1, \cdots , q \}$ of
directional vectors of the extreme rays of $\Lambda_+$ is
known, then it is easy to check, whether this mechanism
satisfies the structural conditions of the extended form of
detailed balance. It is sufficient to examine for each
${\lambda}^s_r>0$, whether $k_r^->0$.

After these conditions are examined, it is a simple task to
extract the independent set of the Wegscheider identities that
should be valid: just select a maximal linear independent
subset from the set of $\boldsymbol{\lambda}^s$ and write the
correspondent Wegscheider identities.

It is convenient to use all the extreme pathways especially if
we would like to study all the subsystems of the given system,
which satisfy the extended form of detailed balance. On the
other hand, it is computationally expensive to find the set
$\{\boldsymbol{\lambda}^s \ | \ s=1, \cdots , q \}$ (see, for
example, the paper by \cite{ElemNodesBMS2004}). The amount of
computation could be significantly reduced because it is not
necessary to use all the extreme pathways.

Let us consider a reaction mechanism, which includes both
reversible and irreversible reactions. For the reversible
reactions, let us join the direct and reverse reactions. Let
$\gamma_1, \ldots , \gamma_r$ be the stoichiometric vectors of
the reversible reactions and $\nu_1, \ldots , \nu_s$ be the
stoichiometric vectors of the irreversible reactions. We use
$\boldsymbol{\Gamma}_r$ for the stoichiometric matrix of the
reversible reactions and $\Lambda_r$ for the solutions of the
equations $\boldsymbol{\lambda}\boldsymbol{\Gamma}_r=0$.

The linear subspace $S={\rm span}\{\gamma_1, \ldots ,
\gamma_r\} \subset \mathbb{R}^n$ consists of all linear
combinations of the stoichiometric vectors of the reversible
reactions. Let us consider the quotient space $\mathbb{R}^n/S$.
We use notation $\overline{\nu}_j$ for the images of $\nu_j$ in
$\mathbb{R}^n/S$.

The following theorem gives the criteria of the extended form
of detailed balance, which are more efficient for computations.

\begin{theorem}\label{Theo:ConstrCrit}The system satisfies the extended form of detailed
balance if and only if
\begin{enumerate}
\item{The convex hull of the stoichiometric vectors of
    irreversible reactions does not intersect $S$, i.e.
    \begin{equation}
    0 \notin {\rm
conv}\{\overline{\nu}_1, \ldots , \overline{\nu}_s\}\, ;
    \end{equation}
    }
\item{The reversible reactions satisfy the principle of
    detailed balance.}
\end{enumerate}
\end{theorem}
\begin{proof}
Let the condition 1 be violated, i.e.  $0 \in {\rm
conv}\{\overline{\nu}_1, \ldots , \overline{\nu}_s\}$. In this case,
there exist such a nonnegative $\theta_i \geq 0$ that $\sum_{j=1}^s
\theta_j =1$ and $\sum_{j=1}^s \theta_j \nu_j \in S$. This means
that $\sum_{j=1}^s \theta_j \nu_j + \sum_{i=1}^r \lambda_i
\gamma_i=0$. We can transform the sum $\sum_{i=1}^r \lambda_i
\gamma_i$ in a combination with positive coefficients if for any
negative $\lambda_i$ we substitute $\gamma_i$ by the stoichiometric
vector of the reverse reaction, that is, $-\gamma_i$. As a result,
we get the element of $\Lambda_+$, a combination of the
stoichiometric vectors with nonnegative coefficients, which is equal
to zero and includes some stoichiometric vectors of the irreversible
reactions with nonzero coefficients. Therefore, the structural
condition of the extended form of detailed balance is violated.

Let the structural condition be violated. Then there exist a
combination $\sum_{j=1}^s \theta_j \nu_j + \sum_{i=1}^r
\lambda_i \gamma_i=0$ with $\theta_j \geq 0$ and $\sum_{j=1}^s
\theta_j>0$. Let us notice that
$$\sum_{j=1}^s
\frac{\theta_j}{\sum_{j=1}^s \theta_j} \nu_j =- \sum_{i=1}^r
\frac{\lambda_i}{\sum_{j=1}^s \theta_j} \gamma_i\, ,$$ and, therefore, $0 \in {\rm
conv}\{\overline{\nu}_1, \ldots , \overline{\nu}_s\}$. The
condition 1 is violated.

We proved that the condition 1 is equivalent to the structural
condition of the extended form of detailed balance.

If the condition 1 holds then the condition 2 is, exactly, the
algebraic condition of the extended form of detailed balance.
\end{proof}

\begin{remark}\label{Rem:Separation}
The first condition of Theorem, $0 \notin {\rm
conv}\{\overline{\nu}_1, \ldots , \overline{\nu}_s\}$, is
equivalent to the existence of such a linear functional $l$ on
$\mathbb{R}^n$ that $l(\nu_j)>0$ for all $j=1,\ldots , s$ and
$l(\gamma_j)=0$ for all $j=1,\ldots , r$.
\end{remark}

\subsection{Linear Systems \label{Sec:linear}}

The results of previous Sections for a linear system
(\ref{Master1}) have a geometrically clear form (see also the
paper by \cite{YabGor2010}).

\begin{proposition}\label{LinNEsSufWeDB}
The necessary and sufficient condition for the extended form of
detailed balance is:  In any cycle $A_{i_1}\to A_{i_2} \to
\ldots \to A_{i_q} \to A_{i_1}$ with the strictly positive
constants $k_{i_{j+1} i_j}>0$ (here $i_{q+1}=i_1$) all the
reactions are reversible ($k_{i_j i_{j+1}}>0$) and the identity
(\ref{WegscheiderCondition}) holds.
\end{proposition}

The states (reagents) $A_q$ and $A_r$ ($q\neq r$) are {\em
strongly connected} if there exist oriented paths both from
$A_q$ to $A_r$ and from $A_r$ to $A_q$ (each oriented edge
corresponds to a reaction with the nonzero reaction rate
constant). From Proposition~\ref{LinNEsSufWeDB} we get the
following statement.

\begin{corollary}
Let a linear system satisfy the extended form of detailed
balance. Then all reactions in any directed path between
strongly connected states are reversible.
\end{corollary}

In brief, a linear system with the extended form of detailed
balance can be described as follows: (i) the oriented cycles
are reversible and satisfy the classical condition
(\ref{WegscheiderCondition}), (ii) the system consists of the
reversible parts and the irreversible transitions between these
parts and (iii) the system of these irreversible transitions is
acyclic.

For example, let us analyze subsystems of the simple cycle,
$A_1 \rightleftharpoons A_2 \rightleftharpoons A_3
\rightleftharpoons A_1$.
\begin{equation}\label{TranspGammaCycle}
\boldsymbol{\Gamma}^{\mathrm{T}}= \left[
\begin{array}{rrrrrr}
 -1 & 0 & 1 & 1 & 0 &-1 \\
 1  &-1 & 0 &-1 & 1 & 0 \\
 0  & 1 &-1 & 0 &-1 & 1
\end{array}
\right]
\end{equation}

The cone of nonnegative solutions $\Lambda_+$ to the equation
$\boldsymbol{\lambda \Gamma}=0$ has extreme rays with the
following direction vectors: $
\boldsymbol{\lambda}^1=(1,1,1,0,0,0)$,
$\boldsymbol{\lambda}^2=(0,0,0,1,1,1)$,
$\boldsymbol{\lambda}^3=(1,0,0,1,0,0)$,
$\boldsymbol{\lambda}^4=(0,1,0,0,1,0)$, and
$\boldsymbol{\lambda}^5=(0,0,1,0,0,1)$. Vectors
$\boldsymbol{\lambda}^{3-5}$ give trivial identities
(\ref{WegscheiderLambda}) $k_i^+k_i^-=k_i^-k_i^+$ ($i=1, 2, 3$)
and vectors $\boldsymbol{\lambda}^{1,2}$ give the same
identity: $k_1^+k_2^+k_3^+=k_1^-k_2^-k_3^-$.

If we delete one elementary reaction from the simple cycle
(i.e. one column from $\boldsymbol{\Gamma}^{\mathrm{T}}$
(\ref{TranspGammaCycle})) then one of the nonnegative solutions
$\boldsymbol{\lambda}^{1,2}$ persists and, due to the extended
detailed balance principle, all the reactions should be
reversible. This means that the structural condition of
extended detailed balance is not satisfied for the simple
reversible cycle without one direct or reverse reaction. If two
reactions are reversible then the third should be reversible or
completely vanish. If we delete one direct reaction (with
number 1, 2 or 3) and one reverse reaction (with number 4, 5 or
6) then there remain no non-trivial solutions in $\Lambda_+$
and, therefore, no non-trivial relations between the constants
persist after deletion of these two reactions.

For the linear systems, the oriented cycles in the graph of
reactions (where vertices are the components and edges are the
reactions) give the positive solutions to the equation
(\ref{lambdaGamma}): for a linear oriented cycle $C$ the sum of the
stoichiometric vectors of its reactions is zero.  Moreover, any
positive solution of (\ref{lambdaGamma}) is a convex combination of
such cyclic solutions and, therefore, the directed vectors of the
extreme rays of $\Lambda_+$ can be selected in this form.

\subsection{Simple Examples of Nonlinear Systems \label{Sec:NonlinExamples}}

In this section, we present several elementary examples. For
these examples, the sets of the extreme pathways are obvious.

Let us examine  a reaction mechanism with irreversible
reactions $A \xrightarrow{k_1} B$  and $2B \xrightarrow{k_2}
2A$.
\begin{equation}\label{GammaFirst}
\boldsymbol{\Gamma}^{\mathrm{T}}= \left[
\begin{array}{rr}
 -1 & 2  \\
 1  &-2
\end{array}
\right]\, .
\end{equation}
The cone $\Lambda_+$ is a ray with the directional vector
$\boldsymbol{\lambda}=(2,1)$. Both $\lambda_{1,2}>0$, hence,
both reactions should be reversible and the condition holds:
$(k_1^+)^2 k_2^+=(k_1^-)^2 k_2^-$.

Let us slightly modify this example: $2{\rm H} \to {\rm H}_2$,
${\rm H}+{\rm H_2} \to 3 {\rm H}$.
\begin{equation}\label{GammaSec}
\boldsymbol{\Gamma}^{\mathrm{T}}= \left[
\begin{array}{rr}
 -2 & 2  \\
 1  &-1
\end{array}
\right]\, .
\end{equation}
The cone $\Lambda_+$ is a ray with the directional vector
$\boldsymbol{\lambda}=(1,1)$. Both $\lambda_{1,2}>0$, hence,
both reactions should be reversible and the condition holds:
$k_1^+ k_2^+=k_1^- k_2^-$.

If we change the direction of one reaction in the previous
example then the new irreversible systems satisfies the
extended form of detailed balance: $2{\rm H} \to {\rm H}_2$, $3
{\rm H} \to {\rm H}+{\rm H_2} $.
\begin{equation}\label{GammaThird}
\boldsymbol{\Gamma}^{\mathrm{T}}= \left[
\begin{array}{rr}
 -2 & -2  \\
 1  &1
\end{array}
\right]\, .
\end{equation}
The cone $\Lambda_+$ is trivial (it includes no rays, just the
origin), hence, the structural condition holds. The algebraic
condition trivially holds, because there is no reversible
reaction.

Let us add  the forth reversible and nonlinear elementary
reaction $A_1+A_2 \rightleftharpoons 2A_3$ (with the constants
$k_4^{\pm}$) to a linear reversible cycle.  We should add to
$\boldsymbol{\Gamma}^{\mathrm{T}}$ (\ref{TranspGammaCycle}) two
new columns:
\begin{equation}\label{TranspGammaCycle2}
\boldsymbol{\Gamma}^{\mathrm{T}}= \left[
\begin{array}{rrrrrrrr}
 -1 & 0 & 1 &-1  & 1 & 0 &-1 & 1 \\
 1  &-1 & 0 &-1  &-1 & 1 & 0 & 1 \\
 0  & 1 &-1 & 2  & 0 &-1 & 1 &-2
\end{array}
\right]
\end{equation}
The extreme rays of $\Lambda_+$ include four rays that
correspond to pairs of mutually reverse reactions
($\boldsymbol{\lambda}^{1-4}$), two rays that correspond to the
linear cycle ($\boldsymbol{\lambda}^{5,6}$) and six rays for
three nonlinear cycles ($\boldsymbol{\lambda}^{7-12}$): (i)
$A_1+A_2 \to 2A_3$, $A_3\to A_2$, $A_3\to A_1$; (ii) $A_1+A_2
\to 2A_3$, $A_3 \to A_1$, $A_1 \to A_2$ and (iii) $A_1+A_2 \to
2A_3$, $A_3 \to A_2$, $A_2 \to A_1$:
\begin{eqnarray*}
&\boldsymbol{\lambda}^5=(1,1,1,0,0,0,0,0), \;\; &\boldsymbol{\lambda}^6=(0,0,0,0,1,1,1,0),\\
&\boldsymbol{\lambda}^7=(0,0,1,1,0,1,0,0),  \;\; &\boldsymbol{\lambda}^8=(0,1,0,0,0,0,1,1),\\
&\boldsymbol{\lambda}^9=(1,0,2,1,0,0,0,0), \;\; &\boldsymbol{\lambda}^{10}=(0,0,0,0,1,0,2,1),\\
&\boldsymbol{\lambda}^{11}=(0,0,0,1,1,2,0,0),\;\;
&\boldsymbol{\lambda}^{12}=(1,2,0,0,0,0,0,1)\, .
\end{eqnarray*}
We omit  $\boldsymbol{\lambda}^{1-4}$ which do not produce
nontrivial conditions. For the reversible reaction mechanism
(when $k^{\pm}_{1-4}>0$), there are two independent Wegscheider
identities (\ref{WegscheiderLambda}) that formalize the
classical principle of detailed balance:
$k_1^+k_2^+k_3^+=k_1^-k_2^-k_3^-$ and
$k_3^+k_4^+k_2^-=k_3^-k_4^-k_2^+$. If some of the elementary
reactions are irreversible then the direction vectors
$\boldsymbol{\lambda}^{5-12}$ produce 8 conditions. For
$\boldsymbol{\lambda}^{5,7,9,11}$ these conditions are below.
\begin{itemize}
\item{($\boldsymbol{\lambda}^5$) If $k_{1,2,3}^+>0$ then
    $k_{1,2,3}^->0$ and $k_1^+k_2^+k_3^+=k_1^-k_2^-k_3^-$;}
\item{($\boldsymbol{\lambda}^7$) If $k_{3,4}^+,k_2^->0$
    then $k_{3,4}^-,k_2^+>0$ and
    $k_3^+k_4^+k_2^-=k_3^-k_4^-k_2^+$;}
\item{($\boldsymbol{\lambda}^9$) If $k_{1,3,4}^+ >0$ then
    $k_{1,3,4}^- >0$ and $k_1^+(k_3^+)^2
    k_4^+=k_1^-(k_3^-)^2 k_4^-$;}
\item{($\boldsymbol{\lambda}^{11}$) If $k_4^+,k_{1,2}^->0$
    then $k_4^-,k_{1,2}^+>0$ and
    $k_4^+k_1^-(k_2^-)^2=k_4^-k_1^+(k_2^+)^2$.}
\end{itemize}
To obtain the conditions for $\boldsymbol{\lambda}^{6,8,10,12}$
it is sufficient to change the superscripts $^+$ to $^-$ and
inverse. These 8 conditions represent the extended form of
detailed balance for a given mechanism. To check, whether a
subsystem of this mechanism satisfies the extended form of
detailed balance, it is necessary and sufficient to check these
conditions.

\subsection{Methane Reforming Processes: a Case Study \label{Sec:Methane}}

\subsubsection{The System}

Methane reforming with CO$_2$ is a complex reaction network
\citep{Benson1981}. The main reactions in the methane reforming
are:
\begin{enumerate}
\item{${\rm CO}_2 + {\rm H_2} \rightleftharpoons  {\rm CO}
    + {\rm H}_2{\rm O}$ (RWGS, Reverse water-gas shift);}
\item{${\rm CH}_4 + {\rm CO}_2 \rightleftharpoons  {\rm
    2CO} + {\rm 2H_2}$ (Dry reforming);}
\item{${\rm CO}_2 + 4{\rm H}_2 \rightleftharpoons  {\rm
    CH}_4 + 2{\rm H}_2{\rm O}$ (Methanation);}
\item{${\rm CH}_4 + {\rm H}_2{\rm O}  \rightleftharpoons  {\rm
    CO} + 3{\rm H}_2$ (Steam reforming);}
\item{${\rm CH}_4 \rightleftharpoons {\rm 2H}_2 + {\rm C}$
    (Methane decomposition);}
\item{$2{\rm CO}  \rightleftharpoons {\rm CO}_2 + {\rm C}$
    (Boudouard reaction);}
\item{${\rm C} + {\rm H}_2{\rm O}  \rightleftharpoons {\rm CO}
    + {\rm H}_2$ (Coal gasification).}
\end{enumerate}
For the reagents, we use the notations $A_1={\rm CH}_4$,
$A_2={\rm CO}_2$, $A_3={\rm CO}$, $A_4={\rm H}_2$, $A_5={\rm
H}_2{\rm O}$, $A_6={\rm C}$. Amount of $A_i$ is $N_i$. There
exist three independent linear conservation laws: $b_{\rm
C}=N_1+N_2+N_3+N_6$; $b_{\rm H}=4N_1+2N_4+2N_5$; $b_{\rm
O}=2N_2+N_3+N_5$. The number of degrees of freedom in the
closed system is three (six components minus three independent
conservation laws).

{This example enriches our discussion because it deviates from the
nice abstract scheme discussed above. First of all, the reactions
1--7 are not elementary steps. We consider them as overall reactions
which have their own intrinsic and complicated reaction mechanism.
This does not cause a serious problem because the generalized mass
action law describes the equilibria of the complex overall reactions
as well as the equilibria of the elementary ones. Therefore, we can
apply the concept of the extended form of detailed balance and our
theorems} \ref{Theo:Equival}--\ref{Theo:ConstrCrit} {to the process
network 1--7 build from the complex reactions. Rigorously speaking,
we deal not with the elementary reaction steps but with the main
equilibria and may discuss, for example, not the ``Boudouard
reaction" but the ``Boudouard equilibrium". }

{The second problem is the heterogeneity of the system: $A_1,
\ldots, A_5$ are gases and $A_6 = {\rm C}$ is solid. Some of the
reactions go on the surface of the solid.

If a multiphase system is ideal and the solid components are
stoichiometric ones (i.e. they have a fixed composition) then
the free energy has the form}
\begin{equation}\label{HeterFree}
F= \sum_{A_i\;   - \; {\rm gas} }N_i (RT\ln c_i + \mu_i^0-RT)+
\sum_{A_i\;   - \; {\rm solid} }N_i \mu_i^0  \, .
\end{equation}
{Here, the free energy of solid components differs from the free
energy of gases by the  absence of the term $RT N \ln c$. This term
corresponds to the ideal gas pressure $PV=NRT$. In our case,}
\begin{equation}\label{HeterFreeC}
F= \sum_{i=1}^5 N_i (RT\ln c_i + \mu_i^0-RT)+ N_6 \mu_6^0  \, .
\end{equation}
{To define the activities, we follow } (\ref{StandardActivity}).
{For the ideal gases $a_i=c_i$ and for the stoichiometric solids
$a_i\equiv 1$.}

{In section} \ref{sec:Detbal}, {we studied homogeneous systems and
considered $x_i=\ln a_i$ as independent unknowns in the detailed
balance equations} (\ref{DetBalGen}):
\begin{equation}\label{DetBalGen2}
\sum_i \gamma_{ri} x_i = \ln K_r \;\;\; (x_i=\ln a_i^{\rm eq})\, .
\end{equation}
{Therefore, for any solution of this system, the activities
$a_i=\exp x_i$ represented a positive equilibrium.

In a heterogeneous system with the free energy} (\ref{HeterFree})
{the activities for the solid components are constant, the
correspondent $x_i\equiv 0$. Let $\boldsymbol{x}=(x_i)$ be a
solution to equations} (\ref{DetBalGen2}), $\boldsymbol{\Gamma
x}=\boldsymbol{K}$, {where $\boldsymbol{K}$ is the vector of the
equilibrium constants. The vector $\boldsymbol{a}=(a_i)$, $a_i=\ln
x_i$ is a vector of equilibrium activities if and only if $x_i=0$
for all the solid components $A_i$. Instead of analyzing the
solvability of the detailed balance equations} (\ref{DetBalGen2}) {
we have to study its solvability under additional condition: $x_i=0$
for all the solid components $A_i$.}

{Let us postpone the discussion of the extended principle of
detailed balance in multiphase systems and consider the system
``gaseous mixture + one stoichiometric solid". Let $A_n$ be
solid.}

{If there is the only solid component then the solvability
conditions for the  system} (\ref{DetBalGen2}) {and for this system
with additional condition $x_n=0$ coincide. Indeed, there exist a
positive stoichiometric linear conservation law:}
$$\sum_{i=1}^n \gamma_{ri} b_i=0 \mbox{  for all  } r \mbox{  and  } b_i>0 \mbox{  for all  } i \, .$$
{For example, this may be conservation of mass or of the amount of
atoms. Let $\boldsymbol{b}=(b_i)$. For any solution of the detailed
balance conditions} (\ref{DetBalGen2}) $\boldsymbol{x}=(x_i)$, the
vector
$$\boldsymbol{x}'=\boldsymbol{x}-
\frac{x_n}{b_n}\boldsymbol{b}$$ {is also a solution to}
(\ref{DetBalGen2}) {with the condition $x'_n=0$.}

{So, for our example with seven equilibria 1--7 the conditions of
the extended principle of detailed balance for the heterogeneous
system with solid $A_6=$C are described by the theorems}
\ref{Theo:Equival}--\ref{Theo:ConstrCrit} {and we can use the
results of the preceding sections.}

\subsubsection{The Classical Wegscheider Conditions}

To formulate the classical Wegscheider identities, we have to
join the direct and inverse reactions and to find the basic
solutions of the system of linear equations
$\boldsymbol{\lambda}\boldsymbol{\Gamma}=0$. The stoichiometric
matrix for this example is:
\begin{equation}\label{TranspGammaCycle5}
\boldsymbol{\Gamma}^{\mathrm{T}}= \left[
\begin{array}{rrrrrrr}
0&  -1   & 1 & -1& -1& 0& 0\\
-1& -1   & -1& 0 & 0 & 1& 0\\
1 &  2   & 0 & 1 & 0 &-2& 1\\
-1&  2   & -4& 3 & 2 & 0& 1\\
1 &  0  & 2 & -1& 0 & 0&-1\\
0 &  0  & 0 &  0& 1 & 1&-1
\end{array}
\right]
\end{equation}
The system of seven equations
$\boldsymbol{\lambda}\boldsymbol{\Gamma}=0$ is redundant. There
are only three independent equations (one equation for every
degree of freedom). It is sufficient to take the components of
stoichiometric vectors that correspond to the components $A_2$,
$A_4$ , $A_6$. Other components satisfy the same linear
relations as the selected ones. The reduced matrix
$\boldsymbol{\Gamma}_{\rm r}^{\mathrm{T}}$ is
\begin{equation}\label{TranspGammaCycle6}
\boldsymbol{\Gamma}_{\rm r}^{\mathrm{T}}= \left[
\begin{array}{rrrrrrr}
-1&-1  & -1& 0 & 0 & 1& 0\\
-1& 2  & -4& 3 & 2 & 0& 1\\
0 & 0  & 0 &  0& 1 & 1&-1
\end{array}
\right]
\end{equation}
There are four independent solutions of the equations
$\boldsymbol{\lambda}\boldsymbol{\Gamma}=0$ (seven variables
minus three independent equations). For example, we can take
the following basis of solutions: $(-1,1,0,-1,0,0,0)$,
$(0,0,0,-1,1,0,1)$, $(1, 0,0,0,0,1,1)$, $(1,0,-1,-1,0,0,0)$.

The correspondent Wegscheider identities are: $K_2=K_1K_4$,
$K_5K_7=K_4$, $K_1K_6K_7=1$, $K_1=K_3K_4$.

\subsubsection{Allowed and Forbidden Mechanisms}

In general, all the seven reactions can be considered as reversible
but under various conditions some of them are almost irreversible.
Let us study which combinations of irreversible reactions are
possible in accordance with the extended form of detailed balance.

For example, existence of the positive solution
$(0,1,0,0,0,1,1) \in \Lambda$ guarantees that the irreversible
system ${\rm CO}_2 + {\rm H_2} \to  {\rm CO}
    + {\rm H}_2{\rm O}$, ${\rm CH}_4 + {\rm CO}_2 \to  {\rm
    2CO} + {\rm 2H_2}$,  ${\rm CO}_2 + 4{\rm H}_2 \to  {\rm CH}_4 +
    2{\rm H}_2{\rm O}$, ${\rm CH}_4 + {\rm H}_2{\rm O}  \to  {\rm
    CO} + 3{\rm H}_2$, ${\rm CH}_4 \to {\rm 2H}_2 + {\rm C}$,
    $2{\rm CO}  \to {\rm CO}_2 + {\rm C}$,
    ${\rm C} + {\rm H}_2{\rm O}  \to {\rm CO}
    + {\rm H}_2$
is forbidden by the extended form of detailed balance. This
conclusion is also obvious from the correspondent Wegscheider
condition $K_2K_6K_7=1$. Indeed, if all the $k_i^- \to 0$ for
bounded from below $k_i^+ > \varepsilon >0$ then all $K_i \to
\infty$ and $K_2K_6K_7 \to \infty$. This contradicts the Wegscheider
condition.

The first reaction (RWGS, Reverse water-gas shift) is reversible in
the wide interval of conditions \citep{Moe1962}. Let us first study
all the reaction mechanisms with the reversible first reaction and
the irreversible reactions 2-7. We find the combinations of the
directions of the irreversible reactions that satisfy the extended
form of detailed balance. As a criterion of the extended form of
detailed balance we use Theorem~\ref{Theo:ConstrCrit}. After that,
we consider other reactions as the reversible ones (in addition to
RWGS) and study the corresponding reaction mechanisms.

The space $S$ is a straight line with the directional vector
$\gamma_1$ with coordinates $(-1,-1,0)$ in the coordinate
system $(N_2,N_4,N_6)$  that corresponds to the components
$A_2$, $A_4$, $A_6$. Let us represent the quotient space
$\mathbb{R}^3/S$ in the coordinate system $(N_2,N_6)$ that
corresponds to the components $A_2$, $A_6$. For this purpose,
we have to eliminate the coordinate $N_4$ using vector
$\gamma_1$. As a result, we get the following vectors:
\begin{equation}\label{quotationVec}
\begin{split}
\overline{\gamma}_2=\left(\begin{array}{r}
-3\\
0
\end{array}\right)\, ,
\overline{\gamma}_3=\left(\begin{array}{r}
3\\
0
\end{array}\right)\, ,
\overline{\gamma}_4=\left(\begin{array}{r}
-3\\
0
\end{array}\right)\, , \\
 \overline{\gamma}_5=\left(\begin{array}{r}
-2\\
1
\end{array}\right)\, ,
 \overline{\gamma}_6=\left(\begin{array}{r}
1\\
1
\end{array}\right)\, ,
 \overline{\gamma}_7=\left(\begin{array}{r}
-1\\
-1
\end{array}\right)\, .
\end{split}
\end{equation}
For example, to find $\overline{\gamma}_2$, we take $\gamma_2$
(the second column in (\ref{TranspGammaCycle6})) and exclude
the coordinate $N_4$ by adding $2\gamma_1$. The result is a
vector $\gamma_2+2\gamma_1$. In coordinates $(N_2,N_6)$, this
vector gives us $\overline{\gamma}_2$.

\begin{figure}
\centering
\includegraphics[width=0.9\textwidth]{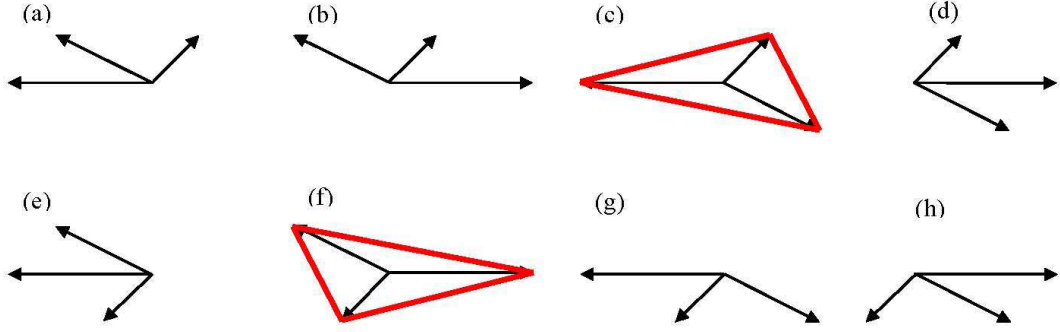}
\caption{\label{convex envelopes}Images of the stoichiometric vectors of
irreversible reactions  $\overline{\nu}_j=\pm \overline{\gamma}_j$ in $\mathbb{R}^3/S$
for various  combinations of
directions of reactions (\ref{directionsCH4}) in coordinates $N_2$ (abscissa), $N_6$. The configurations with $0 \in {\rm
conv}\{\overline{\nu}_2, \ldots , \overline{\nu}_7\}$ are outlined. Vectors $\overline{\nu}_2$,
$\overline{\nu}_3$ and $\overline{\nu}_4
$ coincide as well as vectors $\overline{\nu}_6$ and
$\overline{\nu}_7$.}
\end{figure}

The stoichiometric vectors of irreversible reactions are $+
\gamma_j $ or $- \gamma_j $ ($j=2,\ldots, 7$). Their images in
the quotient space $\mathbb{R}^3/S$ are $+ \overline{\gamma}_j
$ or $- \overline{\gamma}_j $. The extended form of detailed
balance requires that the convex envelope of these vectors
should not include zero. We have to arrange signs in $\pm
\gamma_j $ to provide this property. First of all,  we see
immediately from (\ref{quotationVec}) that the second and the
forth reaction should have the same directions and the third
reaction should have the opposite direction. The directions of
the sixth and the seventh reactions should be opposite.
Therefore, we have to analyze eight possible reaction
mechanisms. Let us represent them by the directions of
reactions:
\begin{equation}\label{directionsCH4}
\left[\begin{array}{c}
({\rm a})\\
\rightleftharpoons\\
\to \\
\leftarrow \\
\to \\
\to\\
\to\\
\leftarrow
\end{array}\right];
\left[\begin{array}{c}
({\rm b})\\
\rightleftharpoons\\
\leftarrow \\
\to \\
\leftarrow \\
\to\\
\to\\
\leftarrow
\end{array}\right];
\left[\begin{array}{c}
({\rm c})\\
\rightleftharpoons\\
\to \\
\leftarrow \\
\to \\
\leftarrow\\
\to\\
\leftarrow
\end{array}\right];
\left[\begin{array}{c}
({\rm d})\\
\rightleftharpoons\\
\leftarrow \\
\to \\
\leftarrow \\
\leftarrow\\
\to\\
\leftarrow
\end{array}\right];
\left[\begin{array}{c}
({\rm e})\\
\rightleftharpoons\\
\to \\
\leftarrow \\
\to \\
\to\\
\leftarrow\\
\to
\end{array}\right];
\left[\begin{array}{c}
({\rm f})\\
\rightleftharpoons\\
\leftarrow \\
\to \\
\leftarrow \\
\to\\
\leftarrow\\
\to
\end{array}\right];
\left[\begin{array}{c}
({\rm g})\\
\rightleftharpoons\\
\to \\
\leftarrow \\
\to \\
\leftarrow\\
\leftarrow\\
\to
\end{array}\right];
\left[\begin{array}{c}
({\rm h})\\
\rightleftharpoons\\
\leftarrow \\
\to \\
\leftarrow \\
\leftarrow\\
\leftarrow\\
\to
\end{array}\right].
\end{equation}
Arrows here correspond to the directions of reactions. For
example, the case (a) corresponds to the reaction mechanism
\begin{enumerate}
\item{${\rm CO}_2 + {\rm H_2} \rightleftharpoons  {\rm CO}
    + {\rm H}_2{\rm O}$;}
\item{${\rm CH}_4 + {\rm CO}_2 \to  {\rm 2CO} + {\rm
    2H_2}$; }
\item{${\rm CO}_2 + 4{\rm H}_2 \leftarrow  {\rm CH}_4 +
    2{\rm H}_2{\rm O}$;}
\item{${\rm CH}_4 + {\rm H}_2{\rm O}  \to {\rm CO} + 3{\rm
    H}_2$;}
\item{${\rm CH}_4 \to {\rm 2H}_2 + {\rm C}$;}
\item{$2{\rm CO}  \to {\rm CO}_2 + {\rm C}$;}
\item{${\rm C} + {\rm H}_2{\rm O}  \leftarrow {\rm CO} +
    {\rm H}_2$.}
\end{enumerate}
Combinations (c) and (f) contradict the condition 1 from
Theorem~\ref{Theo:ConstrCrit}: the origin belongs to the convex
envelope of the vectors $\overline{\nu}_j$ of irreversible
reactions (see Fig.~\ref{convex envelopes}). Hence, only six
combinations of directions of irreversible reactions satisfy
the extended form of detailed balance (from $2^6=64$ possible
combinations of directions): (a), (b), (d), (e), (g) and (h).

Let us extend the list of reversible reactions. If we assume that
the second reaction (dry reforming), is reversible together with the
first one (RWGS) then the third and the forth reactions should be
also reversible because $\gamma_3=2\gamma_1-\gamma_2$ and
$\gamma_4=\gamma_2-\gamma_1$, hence, $\gamma_{3,4}\in {\rm
span}\{\gamma_1,\gamma_2\}$. According to the condition 1 from
Theorem~\ref{Theo:ConstrCrit}, this contradicts to the extended form
of detailed balance if the first and the second reactions are
reversible and the third and the forth are not.

Analogously, in addition to the reversible reaction RWGS, the
{Boudouard equilibrium} 6 and coal gasification 7 can be reversible
only together because $\gamma_7=-\gamma_1-\gamma_6$.

We have to consider three possible sets of reversible
reactions:
\begin{enumerate}
\item{1, 2, 3 and 4;}
\item{1 and 5;}
\item{1, 6 and 7.}
\end{enumerate}

For all three cases,  $\dim S=2$ and $\dim (\mathbb{R}^3/S)=1$.
We will use for the quotient space the coordinate $N_6$ which
corresponds to $A_6={\rm C}$.

In the first case, let us exclude the coordinate $N_4$ from
$\overline{\gamma}_{5,6,7}$ (\ref{quotationVec})  using vector
$\overline{\gamma}_{2}$. We get one-dimensional vectors
$$\overline{\gamma}_{5}=1,\; \overline{\gamma}_{6}=1, \; \overline{\gamma}_{7}=-1\, .$$
To satisfy the extended form of detailed balance the directions of
the fifth and the sixth reaction should coincide and the direction
of the seventh reaction should be opposite: there are two possible
combinations of arrows in irreversible reactions 5, 6 and 7 if
reactions 1, 2, 3 and 4 are reversible: $5 \to,\,6
\to,\,7\leftarrow$ and $5\leftarrow, \,6\leftarrow, \,7\to$.

In the  second case, let us exclude the coordinate $N_4$ from
$\overline{\gamma}_{2,3,4,6,7}$ (\ref{quotationVec})  using
vector $\overline{\gamma}_{5}$. We get one-dimensional vectors
$$\overline{\gamma}_{2}=-2/3,\; \overline{\gamma}_{3}=2/3, \; \overline{\gamma}_{4}=-2/3, \;
\overline{\gamma}_{6}=1/2, \; \overline{\gamma}_{7}=-1/2\, .$$
Again, according to the extended form of detailed balance, here
are two possibilities of directions of irreversible reactions
2, 3, 4, 6 and 7 if reactions 1 and 5 are reversible: $2\to,
\,3\leftarrow, \,4\to, \,6\leftarrow, \,7\to$ and $2\leftarrow,
\,3\to, \,4\leftarrow, \,6\to, \,7\leftarrow$.

In the third case, let us exclude the coordinate $N_4$ from
$\overline{\gamma}_{2,3,4,5}$ (\ref{quotationVec})  using
vector $\overline{\gamma}_{6}$. We get one-dimensional vectors
$$\overline{\gamma}_{2}=3,\; \overline{\gamma}_{3}=-3, \;
\overline{\gamma}_{4}=3, \; \overline{\gamma}_{5}=3\, .$$ According
to the extended form of detailed balance, here are two possibilities
of directions of irreversible reactions 2, 3, 4,  and 5 if reactions
1, 6 and 7 are reversible: $2\to, \,3\leftarrow, \,4\to, \,5\to$ and
$2\leftarrow, \,3\to, \,4\leftarrow, \,5\leftarrow$.

In the first and the third cases, there are nontrivial Wegscheider
identities for the reaction equilibrium constants of reversible
reactions. If reactions 1, 2, 3 and 4 are reversible (case 1) then
$\dim \Lambda=2$ and the basis of $\Lambda$ is, for example,
$\boldsymbol{\lambda}^1=(2,-1,-1,0)$
($2\gamma_1-\gamma_2-\gamma_3=0$) and
$\boldsymbol{\lambda}^2=(1,-1,0,1)$
($\gamma_1-\gamma_2+\gamma_4=0$). The two correspondent Wegscheider
identities are: $K_1^2=K_2K_3$ and $K_1K_4=K_2$ (where
$K_i=k_i^+/k_i^-$).

If the reactions 1, 6 and 7 are reversible then $\dim \Lambda=1$ and
the basis of $\Lambda$ consists of one vector
$\boldsymbol{\lambda}=(1,1,1)$ ($\gamma_1+\gamma_6+\gamma_7=0$). The
correspondent Wegscheider identity is: $K_1K_6K_7=1$.

If we add one more reversible reaction in cases 1-3  then all
the reactions 1-7 should be reversible in according to the
extended form of detailed balance.

In this case study, we demonstrated also how it is possible to
organize computations and reduce the dimension of the
computational problems.

\section{Multiscale Degenerated Equilibria \label{Sec:Shifted}}

Let in a system of reversible reactions with detailed balance
some $k_s^- \to 0$, when the correspondent $k_s^+$ remains
constant and separated from zero. In this case, some
equilibrium activities also tend to zero. Indeed, at
equilibrium $w_s^+=w_s^-$, $w_s^- \to 0$ because $k_s^- \to 0$,
hence, $w_s^+\to 0$ and some of $a_i^{\rm eq}$ with
$\alpha_{si}>0$ also tend to zero due to the generalized mass
action law (\ref{GenMAL}). Therefore, the irreversible limits
of the reactions with detailed balance are closely related to
the limits when some equilibrium activities tend to zero. (For
the usual mass action law is sufficient to replace the
``activity $a_i$" by the ``concentration $c_i$".)

In this section we study asymptotics $a_i^{\rm eq} = {\rm
const}\times \varepsilon^{\delta_i}$, $\varepsilon \to 0$ for
various values of non-negative exponents $\delta_i \geq 0$
($i=1,\ldots,n$).

There exists a well known way to satisfy the principle of
detailed balance: just write $k^-_r=k^+_r/ K_r$ where $K_r$ is
the equilibrium constant:
$$K_r=\frac{\prod_{i=1}^n (a_i^{\rm eq})^{\beta_{ri}}}
{\prod_{i=1}^n (a_i^{\rm eq})^{\alpha_{ri}}}\, .$$
 We can define the equilibrium constant through the equilibrium
 thermodynamics as well (see, for example,
 the classical book by \cite{PrigogineDefay1962}). In this case, the principle of
 detailed balance is also satisfied for the mass action law.

In this approach, we have to group the direct and reverse
reactions together. Therefore, $m$ is here the number of pairs
of reactions, direct + inverse ones. We  deal with $m+n$
constants ($m$ rate constants $k_r^+$ for direct reactions and
$n$ equilibrium data for individual reagents: equilibrium
concentrations or activities) instead of $2m$ constants
$k_r^{\pm}$. For these $m+n$ constants, the principle of
detailed balance produces no restrictions
\citep{GorbanMirFGV1989,YangHlavacek2006}. It holds ``by the
construction" for any positive values of these constants if
$k^-_r=k^+_r/ K_r$ and the equilibrium constants are calculated
in accordance with the equilibrium data.

To transform the conditions of $a_i^{\rm eq} \to 0 $ into
irreversibility of some reactions, it is not sufficient to know
which $a_i^{\rm eq} \to 0 $. We have to take into account the
rates of these convergence to zero for different $i$. In the
simple example, $A_1 \rightleftharpoons A_2 \rightleftharpoons
A_3 \rightleftharpoons A_1$, if $a_{1,2}^{\rm eq} \to 0 $, $a_1
/ a_2 \to 0$ then in the limit we get the system $A_1 \to A_2$
(because the $A_1/A_2$ equilibrium is shifted to $A_2$), $A_1
\to A_3$, $A_2\to A_3$. For the inverse relations between $a_1$
and $a_2$, $a_2 / a_1 \to 0$, the limit system is $A_2 \to A_1$
(the $A_1/A_2$ equilibrium is shifted to $A_1$), $A_1 \to A_3$,
$A_2\to A_3$. For the both limit systems, the equilibrium
activities of $A_1$, $A_2$ are zero but the directions of
reaction are different.

The limit structure of the reaction mechanism when some of
$a_i^{\rm eq} \to 0 $ depends on the behavior of the ratios
$a_i^{\rm eq}/a_i^{\rm eq}$. To formalize this dependence, let
us introduce a parameter  $\varepsilon > 0$ and take $a_i^{\rm
eq} = {\rm const}\times \varepsilon^{\delta_i}$. At
equilibrium, each monomial in the generalized mass action law
is proportional to a power of $\varepsilon$:
$$w_r^{{\rm
eq}+}= k_r^+ {\rm const} \times \varepsilon^{\sum_i \alpha_{ri}
\delta_i}\, , \;\; w_r^{{\rm
eq}-}=k_r^- {\rm const} \times \varepsilon^{\sum_i \beta_{ri}
\delta_i}\, .$$

The principle of detailed balance gives: $w_r^{{\rm
eq}+}=w_r^{{\rm eq}-}$. Therefore,
\begin{equation}\label{asymptk+/k-}
\frac{k_r^+}{k_r^-}={\rm const} \times \varepsilon^{(\gamma_r, \delta)}\, ,
\end{equation}
where $\delta$ is the vector with coordinates $\delta_i$.

There are three possibilities for the reversibility of an
elementary reaction in asymptotic $\varepsilon \to 0$:
\begin{enumerate}
\item{If $(\gamma_r, \delta)=0$ then the reaction remains
    reversible in asymptotic $\varepsilon \to 0$. This
    means that $0<\lim (k_s^+/k_s^-)<\infty$. Therefore, if
    one of the reactions persists in the limit then the
    reverse reaction also persists.}
\item{If $(\gamma_r, \delta)<0$ then in asymptotic
    $\varepsilon \to 0$ can remain only direct reaction.
    This means that $\lim (k_s^{ -}/k_s^{ +})=0$.}
\item{If $(\gamma_r, \delta)>0$ then in asymptotic
    $\varepsilon \to 0$ can remain only reverse reaction.
    This means that $\lim (k_s^{+}/k_s^{-})=0$.}
\end{enumerate}

It is possible that $(\gamma_r, \delta)=0$ but both $k_r^{{\rm
lim}\pm}=0$ just because $k^+_r=0$ and $k^-_r=0$ and not
because of the equilibrium degeneration.  If we delete some
irreversible reactions or several pairs of mutually reverse
reaction then the extended form of detailed balance persists.
Therefore, we do not consider these cases separately and always
discuss the limit reaction mechanisms with $\max\{k_r^{{\rm
lim}+},k_r^{{\rm lim}-}\}>0$.

For each stoichiometric vector $\gamma_r$ the $n$-dimensional
space of vectors $\delta$ is split in three sets: hyperplane
$(\gamma_r, \delta)=0$ (reaction remains reversible), hemispace
$(\gamma_r, \delta)<0$ (only direct reaction remains) and
hemispace $(\gamma_r, \delta)>0$ (only reverse reaction
remains). For the reaction mechanism, intersections of these
sets for all $\gamma_r$ ($r=1,\ldots, m$) form a tiling of the
n-dimensional space of vectors $\delta$. The intersection of
all hyperplanes $(\gamma_r, \delta)=0$ corresponds to the
initial reversible reaction mechanism. Other sets from this
tiling correspond to the reaction mechanisms that are limits of
the initial reaction mechanism when some of the reaction rate
constants tend to zero but the principle of detailed balance is
valid. In our study, the exponents $\delta_j$ should be
non-negative, hence, we have to study the tiling of the
positive orthant $\delta_j \geq 0$ in $\mathbb{R}^n$
Description of the tiling produced by a system of hyperplanes
$(\gamma_r, \delta)=0$ is a classical problem of combinatorial
geometry.

In the usual linear triangle $A_1 \rightleftharpoons A_2
\rightleftharpoons A_3 \rightleftharpoons A_1$ we have to
consider three hyperplanes in the space of exponents
$\delta=(\delta_1,\delta_2,\delta_3)$: $\delta_1=\delta_2$
($(\gamma_1,\delta)=0$), $\delta_2=\delta_3$
($(\gamma_2,\delta)=0$) and $\delta_3=\delta_1$
($(\gamma_3,\delta)=0$). At least one of the exponents should
take zero value to keep the overall concentration in
equilibrium neither zero nor infinite. Let us take
$\delta_1=0$. The hyperplanes turn in the straight lines on the
plane $(\delta_2,\delta_3)$: $0=\delta_2$
($(\gamma_1,\delta)=0$), $\delta_2=\delta_3$
($(\gamma_2,\delta)=0$) and $\delta_3=0$
($(\gamma_3,\delta)=0$). The positive octant on the plane
$(\delta_2,\delta_3)$ is split in five sets (A)--(E), that
correspond to the limits with some irreversible reactions, and
the origin:
\begin{itemize}
\item{(A) $\delta_2=0$, $\delta_3>0$, $A_1
    \rightleftharpoons A_2$, $A_3\to A_1$, $A_3 \to A_2$;}
\item{(B) $\delta_3>\delta_2>0$, $A_3 \to A_2 \to A_1$,
    $A_3 \to A_1$;}
\item{(C) $\delta_3=\delta_2>0$, $A_3 \rightleftharpoons
    A_2$, $A_2 \to A_1$, $A_3 \to A_1$;}
\item{(D) $\delta_2>\delta_3>0$, $A_2 \to A_3 \to A_1$,
    $A_2 \to A_1$ (this case differs from (B) by the
    transposition $2\leftrightarrow 3$);}
\item{(E) $\delta_2>0$, $\delta_3=0$ $A_1
    \rightleftharpoons A_3$ , $A_2\to A_1$, $A_2 \to A_3$
    (this case differs from (A) by the transposition
    $2\leftrightarrow 3$).}
\item{The origin corresponds to the fully reversible
    mechanism.}
\end{itemize}

For a less trivial example, let us analyze the reaction mechanism
from Section~\ref{Sec:NonlinExamples}: $A_1 \rightleftharpoons A_2
\rightleftharpoons A_3 \rightleftharpoons A_1$, $A_1+A_2
\rightleftharpoons 2A_3$. This is a reversible cycle supplemented by
a nonlinear step.

We join the direct and reverse elementary reactions and,
therefore,
\begin{equation}\label{TranspGammaCycle3}
\boldsymbol{\Gamma}^{\mathrm{T}}= \left[
\begin{array}{rrrr}
 -1 & 0 & 1 &-1   \\
 1  &-1 & 0 &-1   \\
 0  & 1 &-1 & 2
\end{array}
\right]
\end{equation}
The columns of this matrix are the stoichiometric vectors
$\gamma_r$.

Let us study the tiling of the positive orthant in
$\mathbb{R}^3$ by the planes $(\gamma_r,\delta)=0$
($r=1,\ldots,4$). First of all, it is necessary and sufficient
to study this tiling of the positive octants in three planes:
$\delta_1=0$, or $\delta_2=0$, or $\delta_3=0$ because at least
one equilibrium concentration should not tend to zero and,
therefore, has zero exponent. The symmetry between $A_1$ and
$A_2$ allows us to study two planes: $\delta_1=0$ or
$\delta_3=0$.

\begin{figure}
\centering
\includegraphics[width=0.5\textwidth]{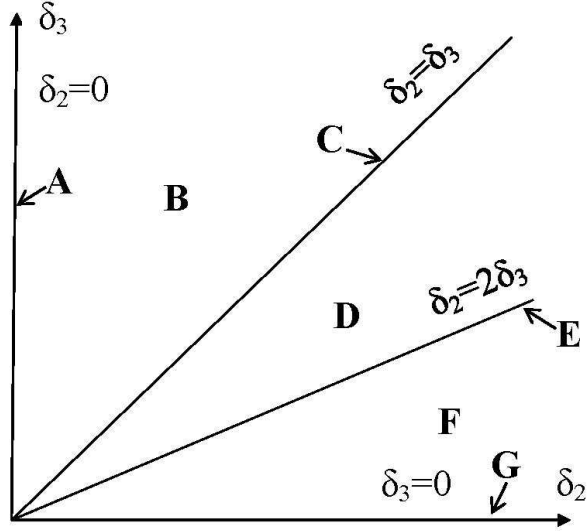}
\caption{\label{Tiling} Tiling of the positive octant of the plane $(\delta_2,\delta_3)$ ($\delta_1=1$)
that corresponds to seven irreversible limits of the reaction mechanism.}
\end{figure}

On the plane $\delta_1=0$ with coordinates $\delta_2$,
$\delta_3$ we have four straight lines: $(\gamma_1,\delta)=0$
($\delta_2=0$), $(\gamma_2,\delta)=0$, ($\delta_2=\delta_3$),
$(\gamma_3,\delta)=0$ ($\delta_3=0$) and $(\gamma_4,\delta)=0$
($\delta_2=2\delta_3$). These lines divide the positive octant
($\delta_{2,3}\geq 0$) into seven parts (Fig.~\ref{Tiling}) and
the origin:
\begin{enumerate}
\item{ (A) $\delta_2=0$, $\delta_3>0$, $A_1
    \rightleftharpoons A_2$, $A_3 \to A_1$, $A_3\to A_2$,
    $2A_3 \to A_1+A_2$;}
\item{(B) $\delta_2>0$, $\delta_3>\delta_2$, $A_2 \to A_1$,
    $A_3 \to A_1$, $A_3\to A_2$, $2A_3 \to
A_1+A_2$;}
\item{(C) $\delta_2=\delta_3>0$, $A_2 \to A_1$, $A_3 \to
    A_1$, $A_3\rightleftharpoons A_2$, $2A_3 \to A_1+A_2$;}
\item{(D) $0<\delta_3< \delta_2<2\delta_3$, $A_2 \to A_1$,
    $A_3 \to A_1$,  $A_2\to A_3$, $2A_3 \to A_1+A_2$;}
\item{(E) $0< \delta_2=2\delta_3$, $A_2 \to A_1$, $A_3 \to
    A_1$,  $A_2\to A_3$, $2A_3 \rightleftharpoons
    A_1+A_2$;}
\item{(F) $\delta_2>2\delta_3>0$, $A_2 \to A_1$, $A_3 \to
    A_1$,  $A_2\to A_3$, $A_1+A_2 \to 2A_3 $;}
\item{(G)  $\delta_3=0$, $\delta_2>0$, $A_2 \to A_1$, $A_1
    \rightleftharpoons A_3$,  $A_2\to A_3$, $A_1+A_2 \to
    2A_3 $;}
\item{The origin corresponds to the fully reversible
    mechanism.}
\end{enumerate}

The same picture gives us the plane $\delta_2=0$ with
coordinates $\delta_1$, $\delta_3$: we need just to transpose
the indexes, $1 \leftrightarrow 2$.

On the plane $\delta_3=0$ with coordinates $\delta_1$,
$\delta_2$ the positive octant is divided into five parts and
the origin:
\begin{enumerate}
\item{$\delta_1=0$, $\delta_2>0$, $A_2 \to A_1$, $A_1
    \rightleftharpoons A_3$, $A_2 \to A_3$, $A_1+A_2 \to
    2A_3$ (this is exactly the case (G) from
    Fig.~\ref{Tiling});}
\item{$0<\delta_1<\delta_2$, $A_2 \to A_1$, $A_1 \to
    A_3$, $A_2 \to A_3$, $A_1+A_2 \to 2A_3$;}
\item{$0<\delta_1=\delta_2$, $A_1 \rightleftharpoons A_2$,
    $A_1 \to A_3$, $A_2 \to A_3$, $A_1+A_2 \to 2A_3$;}
\item{$\delta_1>\delta_2>0$, $A_1 \to A_2$, $A_1 \to A_3$,
    $A_2 \to A_3$, $A_1+A_2 \to 2A_3$;}
\item{$\delta_2=0$, $\delta_1>0$, $A_1 \to A_2$, $A_1 \to
    A_3$, $A_2 \rightleftharpoons A_3$, $A_1+A_2 \to
    2A_3$;}
\item{The origin corresponds to the fully reversible
    mechanism.}
\end{enumerate}

This approach is equivalent to the previous definition of the
extended form of detailed balance based on the pathway
analysis. Indeed, if the reaction mechanism with some
irreversible reactions is a limit of the reversible mechanism
with detailed balance then it satisfies the conditions of the
extended form of detailed balance. (This is the direct
statement of Theorem~\ref{Theo:Equival} proved in
Section~\ref{Sec:Crit}.) To prove the converse statement, we
have to take a system that satisfies the extended form of
detailed balance and to find such a set of exponents $\delta_i
\geq 0$ ($i=1, \ldots, n$) that the system appears in the limit
of a reversible system with detailed balance when
$\varepsilon\to 0$ and $a_i^{\rm eq}={\rm const} \times
\varepsilon^{\delta_i}$.

Let a system with some irreversible reactions satisfy the
extended form of detailed balance. We follow the notations of
Theorem~\ref{Theo:ConstrCrit}: $\gamma_{j}$ ($j=1,\ldots, r$)
are the stoichiometric vectors of the reversible reactions and
$\nu_1, \ldots , \nu_s$ are the stoichiometric vectors of the
irreversible reactions. The linear subspace $S={\rm
span}\{\gamma_1, \ldots , \gamma_r\} \subset \mathbb{R}^n$
consists of all linear combinations of the stoichiometric
vectors of the reversible reactions.  We use notation
$\overline{\nu}_j$ for the images of $\nu_j$ in
$\mathbb{R}^n/S$.

Let $k_j^{\pm}>0$ ($j=1,\ldots , r$)  be the reaction rate
constants for the reversible reactions and $q_j=q_j^{+}>0$
($j=1,\ldots , s$) be the reaction rate constants for the
irreversible reactions. We extend the system by adding the
reverse reactions with the constants $q_j^{-}>0$. If the
extended system satisfies the principle of detailed balance
then
\begin{equation}\label{DetBalExtendExpon}
\frac{k_j^+}{k_j^-}= \prod_{i=1}^n (a_i^{\rm eq})^{\gamma_{ri}}
\;\; \mbox{ and } \; \; \frac{q_j^+}{q_j^-}= \prod_{i=1}^n
(a_i^{\rm eq})^{\nu_{ri}}\, ,
\end{equation}
 where $a_i^{\rm eq}$ is a point of detailed balance.

\begin{theorem}\label{Theo:Exponents+Weak}Let the system satisfy the extended form of detailed
balance. Then there exists a vector of nonnegative exponents
$\delta=(\delta_i)$ ($i=1, \ldots, n$) and the family of
extended systems with equilibria $a_i^{\rm eq}=a_i^*
\varepsilon^{\delta_i}$ such that condition
(\ref{DetBalExtendExpon}) hold, $k_j^{\pm}$ ($j=1,\ldots,r$)
and $q_j=q_j^+$ ($j=1,\ldots, s$) do not depend on
$\varepsilon$, and $q_j^- \to 0$ when $\varepsilon \to 0$.
\end{theorem}
\begin{proof}
If the system satisfies the extended form of detailed balance
then the reversible part satisfies the principle of detailed
balance and, hence, there exists a positive point of detailed
balance for the reversible part of the system
(Theorem~\ref{Theo:ConstrCrit}): $a_i^*>0$ and
$$k_j^+\prod_{i=1}^n (a_i^*)^{\alpha_{ri}}=k_j^+\prod_{i=1}^n
(a_i^*)^{\beta_{ri}}\, .$$
 Let us take $a_i^{\rm eq}=a_i^*\varepsilon^{\delta_i}$.
Due to (\ref{DetBalExtendExpon}), $k_j^+/k_j^-={\rm const}
\times \varepsilon^{(\gamma_j,\delta)}$. To keep the
$k_i^{\pm}$ independent of $\varepsilon$, we have to provide
$(\gamma_j,\delta)=0$. Analogously, $q_j^+/q_j^-={\rm
const}\times  \varepsilon^{(\nu_j,\delta)}$. The rate constant
$q_j^+$ should not depend on $\varepsilon$ and $q_j^- \to 0$
when $\varepsilon \to 0$. Therefore, $(\nu_j,\delta)<0$. We
came to the system of linear equations and inequalities with
respect to exponents $\delta_i$:
\begin{equation}\label{systemForExponents}
(\gamma_j,\delta)=0 \; (j=1, \ldots , r), \;\; (\nu_j,\delta)<0\; (j=1, \ldots , s) \,.
\end{equation}
The solvability of this system is equivalent to the condition 1
of Theorem~\ref{Theo:ConstrCrit} (see
Remark~\ref{Rem:Separation}). To prove the existence of
nonnegative exponents $\delta_i\geq 0$, we have to use
existence of positive conservation law: $b_i>0$,
$(\gamma_j,b)=0$, $(\nu_j,b)=0$. For every solution $\delta$ of
(\ref{systemForExponents}) and any number $d$, the vector
$\delta+db$ is also a solution of (\ref{systemForExponents}).
Therefore, the nonnegative solution exists. We proved the
theorem and the converse statement of
Theorem~\ref{Theo:Equival}.
\end{proof}

\begin{proposition}
Let a system with the stoichiometric vectors $\gamma_s$ and the
extended detailed balance be obtained from the reversible
systems with detailed balance in the limit $a_i^{\rm eq} = {\rm
const}\times \varepsilon^{\delta_i}$, $\varepsilon \to 0$. For
this system, the linear function $(\delta, c)$ of the
concentrations $c$ monotonically decreases in time due to the
kinetic equations $\frac{\D c}{\D t}=\sum_s w_s \gamma_s$.
\end{proposition}
\begin{proof}
Indeed, $\frac{\D (\delta, c)}{\D t}=\sum_s w_s
(\gamma_s,\delta)$ (compare to Remark~\ref{Rem:Acycl}). For the
reversible reactions, the sign of $w_s$ is indefinite but
$(\gamma_s,\delta)=0$. For the irreversible reactions, we
always can take $w_s=w_s^+\geq 0$ just by the selection of
notations. In this case,  only $k^+_s$ survived in the limit
$\varepsilon \to 0$, this means that $(\gamma_s,\delta) <0$.
Therefore, $\frac{\D (\delta, c)}{\D t}\leq 0$ and it is zero
if and only if all the reaction rates of the irreversible
reactions vanish.
\end{proof}

So, the vector of exponents $\delta$ defines the (partially)
irreversible limit of the reaction mechanism and, at the same
time, gives the explicit construction of the special Lyapunov
function for the kinetic equations of the limit system.

In this Section, we developed the approach to the systems with
some irreversible reactions based on multiscale degeneration of
equilibria, when some $a_i \to 0$ as $\varepsilon^{\delta_i}$.
We proved in Theorem~\ref{Theo:Exponents+Weak} that this
approach is equivalent to the extended form of detailed balance
based on the pathways analysis or on the limits of the systems
with detailed balance when some of the reaction rate constants
tend to zero.

\section{Conclusion}

The classical principle of detailed balance operates with
mechanisms, which consist of fully reversible elementary
processes (reactions). If such mechanisms have cycles of
reactions, each cycle is characterized by one Wegscheider
relationship (\ref{WegscheiderLambda}) between its rate
constants. The number of functionally independent relationships
is equal to the number of linearly independent cycles, linear
or nonlinear.

In difference from this classical case, we analyzed mechanisms,
which may include irreversible reactions as well. For such
mechanisms we proved an {\em extended form of detailed
balance} considering the irreversible reactions as limits of
reversible steps, when the rate constants of the corresponding
reverse reactions approach zero. The novelty of this form is
that the extended detailed balance now is presented as a
necessary combination of two constituents:
\begin{itemize}
\item{Structural conditions in accordance to which the
    irreversible reactions cannot be included in oriented
    cyclic pathways.}
\item{Algebraic conditions which are written for the
    ``reversible part" of the complex mechanism taken
    separately, without irreversible reactions, using the
    classical Wegscheider relationships.}
\end{itemize}

The computational tools combine linear algebra (some standard
tools for chemical kinetics) with methods of linear
programming. The most expensive computational problem appears
when we check the structural condition of the extended form of
detailed balance.

Let $n$ be the number of components, and let $\mathbb{R}^n$ be
the composition space. We consider a system with $r$ reversible
and $s$ irreversible reactions. Let us use $\gamma_1, \ldots ,
\gamma_r$ for the stoichiometric vectors of the reversible
reactions, $\nu_1, \ldots , \nu_s$ for the stoichiometric
vectors of the irreversible reactions and $\overline{\nu}_j$
for the images of $\nu_j$ in the quotient space
$\mathbb{R}^n/S$, where $S$ is spanned by the stoichiometric
vectors of all  reversible reaction, $S={\rm span}\{\gamma_1,
\ldots , \gamma_r\} \subset \mathbb{R}^n$. The reaction
mechanism satisfies the structural condition of the extended
form of detailed balance if and only if
$$0 \notin {\rm conv}\{\overline{\nu}_1, \ldots ,
\overline{\nu}_s\}\, .$$

We have to check whether the origin belongs to the convex hull
of the vectors $\overline{\nu}_1, \ldots , \overline{\nu}_s$.
In practice, we can always assume that these vectors have
exactly known rational (or even integer) coordinates.

We combined three approaches to study the restrictions implied by
the principle of detailed balance in the systems with some
irreversible reactions:
\begin{enumerate}
\item{Analysis of limits of the systems with all reversible reactions and
detailed balance when some of the reaction rate constants tend to
zero.}
\item{Analysis of the Wegscheider identities for elementary
    pathways when some of the reaction rate constants turn
    into zero.}
\item{Analysis of limits of the systems when some equilibrium concentrations
(or, more general, activities) tend to zero.}
\end{enumerate}

We proved that these three approaches are equivalent if we take
into account not only which equilibrium concentrations tend to
zero, but the speed of this tendency as well. The various
partially or fully irreversible limits of the reaction
mechanisms are, in this sense, multiscale asymptotics of the
reaction networks when some equilibrium concentration tend to
zero with different speed.

\end{document}